\documentclass[onecolumn, 12pt]{IEEEtran}

\usepackage[T1]{fontenc}
\usepackage{amsthm}
\usepackage{amsmath}
\usepackage[cmintegrals]{newtxmath}
\usepackage{mathtools,amssymb,bm}
\usepackage{subfigure}
\usepackage{graphicx}
\usepackage{color}
\usepackage{multicol}
\usepackage[font={footnotesize}]{caption}
\usepackage{float}

\usepackage{lettrine}
\usepackage{footnote}
\makesavenoteenv{tabular}
\usepackage[noend]{algpseudocode}
\usepackage{hyperref}
\newcommand{\matr}[1]{\mathbf{#1}}

\theoremstyle{definition}
\newtheorem{theo}{\textbf{Theorem}}
\newtheorem{lem}{\textbf{Lemma}}
\newtheorem{cor}{\textbf{Corollary}}
\newtheorem{rem}{\textit{Remark}}

\renewcommand\IEEEkeywordsname{Index Terms}
\renewcommand{\IEEEQED}{$\Box$}

\makeatletter
\newcommand{\Vast}{\bBigg@{5}}
\newcommand{\vast}{\bBigg@{3}}
\makeatother

\newcommand{\dB}{\,\mathrm{dB}}
\newcommand{\E}{\,\mathrm{E}}

\newcommand{\vect}{\,\mathrm{vec}}

\newcommand{\tr}{\,\mathrm{tr}}
\newcommand{\m}{\,\mathrm{m}}
\newcommand{\ddd}{\,\mathrm{d}}

\usepackage{siunitx}
\renewcommand{\baselinestretch}{1.5}
\makeatletter
\providecommand\add@text{}
\newcommand\tagaddtext[1]{%
  \gdef\add@text{#1\gdef\add@text{}}}%
\renewcommand\tagform@[1]{%
  \maketag@@@{\llap{\add@text\quad}(\ignorespaces#1\unskip\@@italiccorr)}%
}
\makeatother

\begin{document}

\title{Jamming-Robust Uplink Transmission for Spatially Correlated Massive MIMO Systems}

\author{\normalsize Hossein~Akhlaghpasand, Emil~Bj\"{o}rnson, and~S.~Mohammad~Razavizadeh
\thanks{H. Akhlaghpasand and S. M. Razavizadeh are with the School of Electrical Engineering, Iran University of Science and Technology, Tehran, Iran (e-mail: h{\_}akhlaghpasand@elec.iust.ac.ir, smrazavi@iust.ac.ir).}
\thanks{E. Bj\"{o}rnson is with the Department of Electrical Engineering, Link\"{o}ping University, Link\"{o}ping, Sweden (e-mail: emil.bjornson@liu.se). }}

\maketitle

\begin{abstract}
In this paper, we consider how the uplink transmission of a spatially correlated massive multiple-input multiple-output (MIMO) system can be protected from a jamming attack. To suppress the jamming, we propose a novel framework including a new optimal linear estimator in the training phase and a bilinear equalizer in the data phase. The proposed estimator is optimal in the sense of maximizing the spectral efficiency of the legitimate system attacked by a jammer, and its implementation needs the statistical knowledge about the jammer's channel. We derive an efficient algorithm to estimate the jamming information needed for implementation of the proposed framework. Furthermore, we demonstrate that optimized power allocation at the legitimate users can improve the performance of the proposed framework regardless of the jamming power optimization. Our proposed framework can be exploited to combat jamming in scenarios with either ideal or non-ideal hardware at the legitimate users and the jammer. Numerical results reveal that using the proposed framework, the jammer cannot dramatically affect the performance of the legitimate system.
\end{abstract}


\section{Introduction}
Over the last decade, massive multiple-input multiple-output (MIMO) has received extensive interests and become the key physical-layer technology for the next generation wireless networks (5G) by providing high-speed data services to a large number of users \cite{BookMarzetta}$-$\cite{CovModEmil}. As the society becomes increasingly reliant on the availability of wireless connectivity, it is of critical importance that future communication networks become robust to jamming. Signal processing at the physical layer can provide a low-cost platform for design of the jamming-robust communication links in massive MIMO systems \cite{KapMag,Ashikh}.

The physical layer processing against jamming can be divided into two categories: Jamming detection and jamming suppression. For jamming detection, we exploit intentionally unused pilots in \cite{JamDetHoss} to propose a detector based on a generalized likelihood ratio test. Kapetanovi\'{c} \emph{et al.} utilize random training for detection of the pilot contamination attack \cite{KapRandomTrain}. In \cite{KapCoopScheme}, cooperative detection methods are proposed where the base station (BS) and a legitimate user cooperate to detect pilot contamination attack. Xie \emph{et al.} define a new frame structure with a two-stages training phase with random intervals for detecting the pilot spoofing attack \cite{TwoStageXie}.

There are some previous papers that provide the first steps towards a jamming-robust transmission \cite{JamMitDet}$-$\cite{AntiJamTransVec}. Vinogradova \emph{et al.} suggest an approximate minimum mean-squared error (MMSE) combiner based on random matrix theory to reject the jamming \cite{JamMitDet}. In \cite{JamResTai}, a receiver is proposed that exploits channel estimates of both legitimate user and jammer to improve the spectral efficiency (SE). In \cite{JamSupHoss}, we propose a framework including a novel MMSE based jamming suppression (MMSE-JS) estimator for channel training and a zero-forcing jamming suppression (ZFJS) detector for uplink combining.  Zhao \emph{et al.} suggest two anti-jamming algorithms based on cooperation between the transmitter and receiver with perfect channel state information \cite{AntiJamTransVec}.

An active eavesdropper can be viewed as a jammer that transmits disturbing signals only in the pilot phase for manipulating channel estimation in order to increase accuracy of its wiretap. Xiong \emph{et al.} in \cite{TwoWayXiong} maximize the secrecy rate in the presence of an active eavesdropper with a frame structure composing two training phases. Wang \emph{et al.} in \cite{SecTransWang} assign the beamforming coefficients such that the approximate secrecy rate is maximized, while in \cite{ArtificialNoiseWu}, the asymptotic secrecy rate is derived assuming matched filter beamforming and artificial noise generation at the BS. Optimal power control for preventing eavesdropping in the downlink transmission of the cell-free massive MIMO systems is proposed in \cite{SecureCellFreeNgo} for ideal hardware and in \cite{SecureCellFreeZhang} for non-ideal hardware. Previous works have either considered security problem in the presence of active eavesdropper \cite{TwoWayXiong}$-$\cite{SecureCellFreeZhang} or jamming suppression problem for only uncorrelated Rayleigh fading channels \cite{JamMitDet}$-$\cite{AntiJamTransVec}. In general, the jammer's strategy depends on its objective and its available information. Accordingly, in \cite{Ashikh}, \cite{JamDetHoss}, \cite{JamResTai}, \cite{JamSupHoss}, \cite{SecTransWang}, the authors suppose that the pilot jamming attack is a combination of several pilot sequences of the pilot set, while in \cite{KapMag}, \cite{KapRandomTrain}$-$\cite{JamMitDet}, \cite{TwoWayXiong}, \cite{ArtificialNoiseWu}$-$\cite{SecureCellFreeZhang}, the jammer selects a target user's pilot sequence as the jamming pilot.

In this paper, we propose a new framework for suppressing the jamming in the uplink transmission of the massive MIMO systems with spatially correlated channels. The utilization of spatial correlation for jamming suppression is a new dimension that has not been explored before in the literature. The contributions of this paper can be summarized as follows:

\begin{itemize}
  \item We introduce a new jamming-robust framework including a SE-maximizing linear estimator and a bilinear equalizer (BE). The proposed framework can perform well in the spatially correlated channels, which is more applicable in practice.
  \item We show that the performance of the proposed framework improves when the similarty between the statistical knowledges of the jammer's channel and the legitimate users' channels decreases. The other advantage of the proposed framework is acceptable performance in the high jamming power regime.
  \item We demonstrate that to implement the proposed framework in the first item, we need statistical knowledge about the jammer's channel. Thus, we derive an algorithm to estimate this knowledge by exploiting the received signals over some consecutive coherence blocks.
  \item In the proposed framework, we study how the legitimate user (jammer) can optimally allocate its power budget to the pilot and data phases to maximize (minimize) the performance of the legitimate system. For jamming suppression in the spatially correlated channels, the problem of power allocation has not been studied before in the literature.
  \item We also study the problem of jamming suppression with non-ideal transceiver hardware. To the best of our knowledge, this problem also has not been considered in prior works. The effect of the distortions from hardware impairments of the legitimate users and the jammer on the performance of the proposed framework is investigated.
\end{itemize}
Numerical results demonstrate that massive MIMO systems equipped with the proposed framework are more robust against jamming attacks than a framework consisting of the MMSE estimator and the zero-forcing (ZF) detector. The simulations also indicate that the proposed framework performs well even if the legitimate users and the jammer operate with non-ideal hardware.

\textit{Notations:} The conjugate, transpose, and conjugate transpose of an arbitrary matrix $\matr{X}$ are respectively denoted by $\matr{X}^*$, $\matr{X}^T$, and $\matr{X}^H$, while $\tr \left(\matr{X}\right)$ represents the trace of the matrix $\matr{X}$. The Kronecker product of two matrices $\matr{X}$ and $\matr{Y}$ is denoted by $\matr{X} \otimes \matr{Y}$. The operator $\vect \left(\matr{X}\right)$ yields the vector with the vectorization elements of the matrix $\matr{X}$. $\matr{I}_N$ is the $N \times N$ identity matrix and $\matr{0}$ is the all-zero column vector.

\section{Uplink Massive MIMO Setup}
We consider the uplink of a single-cell massive MIMO system, depicted in Fig. \ref{fig1}, containing one $M$-antenna BS and $K$ single-antenna legitimate users (hereafter called \emph{users}). We assume that a single-antenna jammer attacks a target user (e.g., the $m$th user) in order to reduce its SE. The model can be generalized to the case that the jammer attacks all users. We denote by $\matr{h}_{k} \in \mathbb{C}^{M \times 1},~k=1,\ldots,K$ the channel vector from the $k$th user to the BS, which is modeled by $\matr{h}_{k} \sim \mathcal{CN} \left(\matr{0} , \matr{R}_k \right)$, where $\matr{R}_k = \text{E}\{\matr{h}_k \matr{h}_k^H \} \in \mathbb{C}^{M \times M}$ is the positive definite channel covariance matrix of the $k$th user. Moreover, $\matr{h}_{w} \in \mathbb{C}^{M \times 1}$ denotes the channel vector from the jammer to the BS, where $\matr{h}_{w} \sim \mathcal{CN} \left(\matr{0} , \matr{R}_w \right)$ and $\matr{R}_w = \text{E}\{\matr{h}_w \matr{h}_w^H \} \in \mathbb{C}^{M \times M}$ is the positive definite channel covariance matrix of the jammer. We consider a block-fading model where the channels are fixed within a coherence block of $T$ samples and take independent realizations in each block.

The users transmit their mutually orthonormal pilots to the BS at $\tau$ ($K \leq \tau \leq T$) samples of a coherence block in order to perform the channel estimation. During the remaining $T-\tau$ samples, the users transmit data symbols to the BS. We denote by $p_t$ and $p_d$ the average transmit powers of the users over the pilot and data phases, respectively. To analyze the worst-case impact of the jamming, we assume that the jammer is smart and knows the transmission protocol. It uses different transmit powers $q_t$ and $q_d$ during jamming the pilot and data phases, respectively. These powers can be optimized to make the jamming as effective as possible in reducing the SE.

\subsection{Pilot Phase}
The pilots of the users are selected from an orthonormal pilot set $\left \{\boldsymbol{\phi}_1, \ldots, \boldsymbol{\phi}_{\tau} \right \}$ in which $\boldsymbol{\phi}_i \in \mathbb{C}^{\tau \times 1}$ is the $i$th pilot and

\begin{equation}
\boldsymbol{\phi}_i^{T} \boldsymbol{\phi}_l^{*}=
\begin{dcases*}
\begin{aligned}
1 , ~~ & l=i , \\
0 , ~~ & l \neq i .
\end{aligned}
\end{dcases*}
\end{equation}

We assume that a smart jammer is aware of the pilot sequence assigned to the target user (i.e., the $m$th user).

\begin{figure}
\centering
\includegraphics[width=3.55in]
{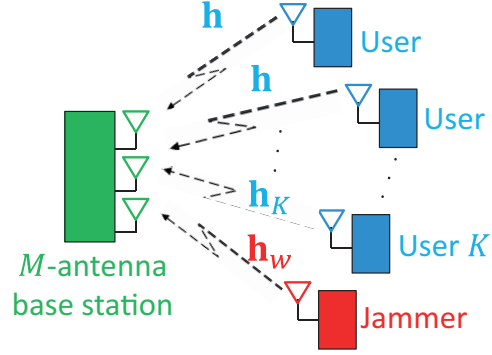}
\caption{Massive MIMO uplink scenario under a jamming attack.}
\label{fig1}
\end{figure}

\subsubsection{Pilot Jamming Attack}
The jammer selects the $m$th user's pilot sequence as the jamming pilot in order to reduce that user's channel estimation quality and thereby limit the SE of the user.

\subsubsection{Channel Estimation}
In the pilot phase, the users send their allocated pilots and the jammer transmits the $m$th user's pilot to create intentional pilot contamination. The received signal $\matr{Y\!}_t \in \mathbb{C}^{M \times \tau}$ at the BS is given by

\setcounter{equation}{1}
\begin{equation}\label{received_matrix}
\matr{Y\!}_t = \sqrt{\tau p_t} \sum_{i=1}^{K}\matr{h}_{i} \boldsymbol{\phi}^T_{i} + \sqrt{\tau q_t} \matr{h}_w \boldsymbol{\phi}^T_{m} + \matr{N}_t ,
\end{equation}
where $\matr{N}_t \in \mathbb{C}^{M \times \tau}$ is the normalized receiver noise matrix composed of i.i.d. $\mathcal{CN}(0,1)$ elements. The projection of $\matr{Y\!}_t$ on the $k$th user's pilot sequence is

\begin{equation}\label{proj_est_act}
\matr{y}_{k} = \matr{Y\!}_t \boldsymbol{\phi}_{k}^* = 
\begin{dcases*}
\begin{aligned}
&\sqrt{\tau p_t} \matr{h}_{m} + \sqrt{\tau q_t} \matr{h}_w + \matr{n}_{m} , ~ && k=m , \\
&\sqrt{\tau p_t} \matr{h}_{k} + \matr{n}_{k} , ~ && k \neq m ,
\end{aligned}
\end{dcases*}
\end{equation}
where $\matr{n}_{k} = \matr{N}_{t} \boldsymbol{\phi}_{k}^* \sim \mathcal{CN} \left( \matr{0} , \matr{I}_M \right)$. Any linear estimate of the channel $\matr{h}_k$ based on $\matr{y}_k$ is obtained as

\begin{equation}\label{leg_est}
\hat{\matr{h}}_{k} = \matr{A}_k \matr{y}_k ,
\end{equation}
where $\matr{A}_k \in \mathbb{C}^{M \times M}$ is a deterministic matrix dependent on the channel statistics. The estimate $\hat{\matr{h}}_{k}$ is a zero-mean random vector with covariance matrix $\matr{A}_k \matr{B}_{k} \matr{A}_k^H$, where $\matr{B}_{k}$ is the covariance matrix of $\matr{y}_k$:

\begin{equation}\label{cov_LS_est}
\matr{B}_k \triangleq \text{E} \left\{\matr{y}_k \matr{y}_k^H \right\} = 
\begin{dcases*}
\begin{aligned}
&\tau p_t \matr{R}_m + \matr{S}_w^{\left(t\right)} + \matr{I}_M , ~ && k=m , \\
&\tau p_t \matr{R}_k + \matr{I}_M , ~ && k \neq m ,
\end{aligned}
\end{dcases*}
\end{equation}
and $\matr{S}_w^{\left(t\right)} \triangleq \tau q_t \matr{R}_w$. Unlike the MMSE method \cite{BookEstKay} that determines $\matr{A}_k$ to minimize the mean-squared error (MSE), we will discuss later how to calculate $\matr{A}_k$ to suppress the interference caused by the jammer.

\subsection{Data Phase}
The normalized complex symbols $x_{i}$ and $x_w$ are simultaneously transmitted by the $i$th user and the jammer during the data phase. The received data signal $\matr{y}_d \in \mathbb{C}^{M \times 1}$ at the BS is given by

\begin{equation}
\label{rec_sig}
\matr{y}_d = \sqrt{p_d} \sum_{i=1}^{K} \matr{h}_{i} x_{i} + \sqrt{q_d} \matr{h}_w x_w + \matr{n}_d ,
\end{equation}
where $\matr{n}_d \in \mathbb{C}^{M \times 1}$ is the normalized receiver noise distributed as $\matr{n}_d \sim \mathcal{CN} (\matr{0} , \matr{I}_M )$. The covariance matrix of $\matr{y}_d$ is

\begin{equation}\label{cov_mat}
\matr{Q} \triangleq \E \left\{\matr{y}_d \matr{y}_d^H \right\} = p_d \sum_{i=1}^{K} \matr{R}_i + \matr{S}_w^{\left(d\right)} + \matr{I}_M ,
\end{equation}
where $\matr{S}_w^{\left(d\right)} \triangleq q_d \matr{R}_w$.
The BS filters the data signal $\matr{y}_d$ with a linear detector $\matr{v}_k$ to detect the signal of the $k$th user:

\begin{equation}
\label{det_sig_imperfect}
\hat{x}_k = \matr{v}_k^H \matr{y}_d = \sqrt{p_d} \sum_{i=1}^{K} \matr{v}_k^H \matr{h}_{i} x_{i} + \sqrt{q_d} \matr{v}_k^H \matr{h}_w x_w + \matr{v}_k^H \matr{n}_d .
\end{equation}
The BE based on the channel estimate in \eqref{leg_est}, $\matr{v}_k = \hat{\matr{h}}_k$, is a low-complexity detector that is known as a good match for massive MIMO systems \cite{Neumann}. Since BE provides a linear processing with respect to the data signal $\matr{y}_d$ and also the pilot signal $\matr{y}_k$ in \eqref{leg_est}, i.e., $\hat{x}_k = \matr{y}_k^H \matr{A}_k^H \matr{y}_d$, we can design the matrix $\matr{A}_k$ such that the SE of the $k$th user in the presence of jamming is maximized. An identical expression for \eqref{det_sig_imperfect} in the case of using the BE detector is

\begin{equation}
\label{det_sig_imperfect2}
\hat{x}_{k} = \sqrt{p_d} \text{E} \left\{\hat{\matr{h}}_k^H \matr{h}_{k} \right\} x_{k} + \hat{n}_{k} ,
\end{equation}
where $\hat{n}_{k}$ represents the interferences plus noise and is given by

\begin{equation}\label{noise_imperfect}
\hat{n}_{k} \triangleq \sqrt{p_d} \sum_{i=1}^{K} \hat{\matr{h}}_k^H \matr{h}_{i} x_{i} - \sqrt{p_d} \text{E} \left\{\hat{\matr{h}}_k^H \matr{h}_{k} \right\} x_{k} + \sqrt{q_d} \hat{\matr{h}}_k^H \matr{h}_w x_w + \hat{\matr{h}}_k^H \matr{n}_d .
\end{equation}
The desired signal from the $k$th user is decoded by treating $\hat{n}_{k}$ as noise \cite{BookMarzetta}, leading to the following achievable SE.

\begin{lem} \label{Lemma1}
Using the BE detector with correlated Rayleigh fading, an achievable SE of the $k$th user is obtained as

\begin{equation}
\label{erg_cap_imperfect}
\mathcal{S}_k = \left(1-\frac{\tau}{T}\right) \log_2 \left(1+{\rho}_{k} \right) ,
\end{equation}
where ${\rho}_{k}$ represents the effective signal-to-interference-plus-noise ratio (SINR) of the $k$th user and is given by

\begin{equation}\label{sjnr_imperfect}
{\rho}_{k} = 
\begin{dcases*}
\begin{aligned}
&\frac{\tau p_t p_d \left|\tr \left(\matr{A}_m^H \matr{R}_m \right) \right|^2}{\tau q_t q_d \left| \tr \left(\matr{A}_m^H \matr{R}_w \right) \right|^2 + \tr \left(\matr{A}_m^H \matr{Q} \matr{A}_m \matr{B}_m \right)} , ~ && k=m , \\
&\frac{\tau p_t p_d \left|\tr \left(\matr{A}_k^H \matr{R}_k \right) \right|^2}{\tr \left(\matr{A}_k^H \matr{Q} \matr{A}_k \matr{B}_k \right)} , ~ && k \neq m .
\end{aligned}
\end{dcases*}
\end{equation}
\end{lem}

\begin{proof}
Please refer to Appendix \ref{proof of Lemma1}.
\end{proof}
\begin{rem}
\it
Comparing the effective SINRs in \eqref{sjnr_imperfect} for two cases $k=m$ and $k \neq m$, the jammer decreases the effective SINR of the $m$th user (i.e., the target user) more than the effective SINRs of the other users, since it creates the contamination only on the pilot of the target user.
\end{rem}
The first term in the denominator of the $m$th user's SINR in \eqref{sjnr_imperfect}, which is caused by the pilot contamination, is similar to the desired term in the numerator. Therefore, to reduce it without decreasing the desired term, the covariance matrices $\matr{R}_m$ and $\matr{R}_w$ should be linearly independent. In an uncorrelated fading channel, the covariance matrices are scaled identity matrices and the above independency is not satisfied. However, in spatially correlated fading channels, the covariance matrices can be linearly independent. Hence, utilization of the spatial correlation provides an opportunity for design of a processing to reduce the pilot contamination caused by the jammer.

\section{Robust Uplink Transmission Against Jamming Attack}
In this section, we discuss how to design the matrix $\matr{A}_k$ in \eqref{leg_est}, in order to reduce the jammer's effect in the SINR in \eqref{sjnr_imperfect}.

\subsection{Jamming-robust Channel Estimation}
The proposed approach for design of an optimal estimator is determining $\matr{A}_k$ such that the effective SINR in \eqref{sjnr_imperfect} is maximized. We call this approach as ``maximum SE (MS)''.
\begin{theo}\label{theo1}
The optimal vector $\matr{a}_k^{\star} \triangleq \vect (\matr{A}_k^{\star}) \in \mathbb{C}^{M^2 \times 1}$ which maximizes \eqref{sjnr_imperfect} can be obtained as
\setcounter{equation}{12}

\begin{equation}\label{opt_sjnr2}
\matr{a}_k^{\star} = 
\begin{dcases*}
\begin{aligned}
&\left(\matr{B}_{m}^T \otimes \matr{Q} + \matr{s}_w^{\left(t\right)} \matr{s}_w^{\left(d\right)^H} \right)^{-1} \matr{r}_m , ~ && k=m , \\
&\left(\matr{B}_{k}^T \otimes \matr{Q} \right)^{-1} \matr{r}_k , ~ && k \neq m ,
\end{aligned}
\end{dcases*}
\end{equation}
where $\matr{r}_k = \vect (\matr{R}_k) \in \mathbb{C}^{M^2 \times 1}$, $\matr{s}_w^{\left(t\right)} = \vect (\matr{S}_w^{\left(t\right)}) \in \mathbb{C}^{M^2 \times 1}$ and $\matr{s}_w^{\left(d\right)} = \vect (\matr{S}_w^{\left(d\right)}) \in \mathbb{C}^{M^2 \times 1}$.
\end{theo}
\begin{proof}

Please refer to Appendix \ref{proof of Theorem1}.
\end{proof}
By using the optimal estimator, we can explicitly compute the optimal effective SINR as

\begin{equation}\label{opt_sjnr_r2}
\rho_k^{\star} = 
\begin{dcases*}
\begin{aligned}
&\tau p_t p_d \matr{r}_m^H \left(\matr{B}_{m}^T \otimes \matr{Q} + \matr{s}_w^{\left(t\right)} \matr{s}_w^{\left(d\right)^H} \right)^{-1} \matr{r}_m , ~ && k=m , \\
&\tau p_t p_d \matr{r}_k^H \left(\matr{B}_{k}^T \otimes \matr{Q} \right)^{-1} \matr{r}_k , ~ && k \neq m .
\end{aligned}
\end{dcases*}
\end{equation}

\subsection{Jamming Channel Statistics}\label{SecIII-B}
To implement the MS for the target user, we need to know the jamming channel statistics $\matr{S}_w^{\left(t\right)}$ and $\matr{S}_w^{\left(d\right)}$. To estimate $\matr{S}_w^{\left(t\right)}$, we can exploit the pilot signal $\matr{y}_m = \matr{Y}_t \boldsymbol{\phi}_m^*$ over $N$ coherence blocks denoted by $\matr{y}_m \left[1\right],\ldots, \matr{y}_m \left[N\right]$ for the $m$th pilot signal and obtained from the pilot already used for channel estimation.

Regarding \eqref{cov_LS_est}, $\matr{S}_w^{\left(t\right)}$ is obtained as

\begin{equation}\label{JamChSt}
\matr{S}_w^{\left(t\right)} = \matr{B}_m -\tau p_t \matr{R}_m-\matr{I}_M .
\end{equation}
To compute $\matr{S}_w^{\left(t\right)}$ in \eqref{JamChSt}, all variables except $\matr{B}_m$ are known or estimated before the jamming attack. Hence, we need to estimate $\matr{B}_m$ in order to obtain $\matr{S}_w^{\left(t\right)}$ which is estimated from the law of large numbers and the ergodicity of the channels:

\begin{equation}\label{EstVar_b_i}
\hat{\matr{B}}_m = \frac{1}{N} \sum_{n=1}^{N} \matr{y}_m \left[n\right] \left(\matr{y}_m \left[n\right] \right)^H .
\end{equation}
Finally, using \eqref{JamChSt} and \eqref{EstVar_b_i}, we can estimate $\matr{S}_w^{\left(t\right)}$ as

\begin{equation}\label{EstVar_s_w}
\hat{\matr{S}}_w^{\left(t\right)} =  \hat{\matr{B}}_m-\tau p_t \matr{R}_m-\matr{I}_M .
\end{equation}

In addition, to estimate $\matr{S}_w^{\left(d\right)}$, the BS knows the users' covariance matrices and their transmit powers in the data phase, thus from \eqref{cov_mat}, we have

\begin{equation}\label{Est_Rec_Jam_Pow1}
\hat{\matr{S}}_w^{\left(d\right)} = \hat{\matr{Q}}-p_d \sum_{i=1}^{K} \matr{R}_i-\matr{I}_M .
\end{equation}
The BS exploits the data signal $\matr{y}_d$ over $N$ coherence blocks to estimate $\hat{\matr{Q}}$ as

\begin{equation}\label{Est_B_d}
\hat{\matr{Q}} = \frac{1}{N} \sum_{n=1}^{N} \matr{y}_d \left[n\right] \left(\matr{y}_d \left[n\right] \right)^H .
\end{equation}
The standard deviation of the sample estimates in \eqref{EstVar_b_i} and \eqref{Est_B_d} is proportional to $1/\sqrt{N}$, therefore even with a relatively small $N$, we can get good estimates.

\begin{figure*}[!t]
\setcounter{equation}{33}
\begin{multline}\label{sjnr_imperfect_HardImpair}
\tilde{\rho}_{k} = \\
\begin{dcases*}
\begin{aligned}
&\frac{\tau p_tp_d\left(1-\kappa_u^2 \right)^2 \left| \tr \left(\matr{A}_m^H \matr{R}_m \right) \right|^2}{q_t q_d\left(\tau -\left(\tau -1\right)\kappa_w^2 \right) \left|\tr \left(\matr{A}_m^H \matr{R}_w \right) \right|^2 + p_tp_d\kappa_u^2 \sum_{i=1}^{K} \left|\tr \left(\matr{A}_m^H \matr{R}_i \right)\right|^2 + \tau p_tp_d\kappa_u^2 \left(1-\kappa_u^2 \right) \left|\tr \left(\matr{A}_m^H \matr{R}_m \right)\right|^2 + \tr \left(\matr{A}_m^H \matr{Q} \matr{A}_m \tilde{\matr{B}}_{m} \right)} , && k=m \\
&\frac{\tau p_tp_d\left(1-\kappa_u^2 \right)^2 \left| \tr \left(\matr{A}_k^H \matr{R}_k \right) \right|^2}{q_t q_d\kappa_w^2 \left|\tr \left(\matr{A}_k^H \matr{R}_w \right) \right|^2 + p_tp_d \kappa_u^2 \sum_{i=1}^{K} \left|\tr \left(\matr{A}_k^H \matr{R}_i \right)\right|^2 + \tau p_tp_d\kappa_u^2 \left(1-\kappa_u^2 \right) \left|\tr \left(\matr{A}_k^H \matr{R}_k \right)\right|^2 + \tr \left(\matr{A}_k^H \matr{Q} \matr{A}_k \tilde{\matr{B}}_{k} \right)} , && k \neq m
\end{aligned}
\end{dcases*}
\end{multline}
\hrulefill
\end{figure*}

\section{Optimal Power Allocation}
We consider two main questions as follows: From the system's point of view, how should each user allocate its power budget to the pilot and data phases in order to maximize effective SINR $\rho_m^{\star}$?, and from the jammer's point of view, how does the jammer allocate its power budget to the pilot and data phases in order to minimize $\rho_m^{\star}$?

The power budgets for each user and the jammer are denoted by $\mathcal{P}$ and $\mathcal{P}_w$ \cite{SubvertHess}, respectively, where the constraint on the transmit powers of each user is

\setcounter{equation}{19}
\begin{equation}\label{PowBudg_Eq}
\tau p_t+\left(T-\tau\right)p_d = \mathcal{P}T ,
\end{equation}
and the transmit powers of the jammer satisfy

\begin{equation}\label{PowBudg_Eq2}
\tau q_t+\left(T-\tau\right)q_d = \mathcal{P}_wT .
\end{equation}

\subsection{System's Standpoint}
From the system's point of view, the optimal values $p_t^{\text{o}}$ and $p_d^{\text{o}}$ are derived from the following optimization problem

\begin{equation}\label{optimization user0}
\mathfrak{L}:
\begin{dcases*}
\begin{aligned}
& \underset{p_t, p_d}{\text{maximize}}
&& \rho_m^{\star} \\
& \text{subject to} && \tau p_t+\left(T-\tau\right)p_d = \mathcal{P}T, \\
& && p_t \geq 0, p_d \geq 0, \\
\end{aligned}
\end{dcases*}
\end{equation}
where the jamming channel statistics $\matr{S}_w^{\left(t\right)}$ and $\matr{S}_w^{\left(d\right)}$ can be acquired using the sample estimates in Section \ref{SecIII-B}. We cannot generally show that the objective function in $\mathfrak{L}$ is concave, however the optimal values $p_t^{\text{o}}$ and $p_d^{\text{o}}$ can be found numerically. This should be taken into consideration that the relaxation methods, such as those proposed by \cite{SecureCellFreeNgo}, \cite{SecureCellFreeZhang} for the uncorrelated fading channels, are impractical in our setup with the correlated fading channels. Moreover, we can analytically obtain useful insights into $p_t^{\text{o}}$ and $p_d^{\text{o}}$ in one interesting scenario that is $\mathcal{P}_w \gg \mathcal{P} \geq 1$.

\begin{cor}\label{cor1}
In the case of $\mathcal{P}_w \gg \mathcal{P} \geq 1$, the solution for $\mathfrak{L}$ is given by

\begin{equation}\label{opt_sol_leg_sys}
\begin{dcases*}
\begin{aligned}
&p_t^{\text{o}}=\frac{\mathcal{P}T}{2\tau}, \\
&p_d^{\text{o}}=\frac{\mathcal{P}T}{2\left(T-\tau \right)}.
\end{aligned}
\end{dcases*}
\end{equation}
\end{cor}

\begin{proof}
Please refer to Appendix \ref{proof of Corollary1}.
\end{proof}

For a given coherence interval $T$ and the power budget $\mathcal{P}$, the optimal transmit powers for each user depend on the pilot length $\tau$, so that

\begin{equation}\label{Res_Pow_Alloc}
\tau p_t^{\text{o}} = \left(T-\tau\right)p_d^{\text{o}} = \frac{\mathcal{P}T}{2} .
\end{equation}

This means that the optimal power allocation is equally dividing the power budget between the pilot and data phases against a high power jammer.

\subsection{Jammer's Standpoint}
Now, we consider the power allocation from the jammer's point of view. The smart jammer can acquire the channel statistics, $p_t$ and $p_d$ to impose the maximum loss to the SE of the target user. In this case, the optimization problem is

\begin{equation}\label{optimization user1}
\mathfrak{J}:
\begin{dcases*}
\begin{aligned}
& \underset{q_t, q_d}{\text{minimize}}
&& \rho_m^{\star} \\
& \text{subject to} && \tau q_t+\left(T-\tau\right)q_d = \mathcal{P}_wT, \\
& && q_t \geq 0, q_d \geq 0. \\
\end{aligned}
\end{dcases*}
\end{equation}

We use the next theorem to efficiently solve $\mathfrak{J}$.
\begin{theo}\label{theo3}
$\mathfrak{J}$ is a convex optimization problem.
\end{theo}

\begin{proof}
Please refer to Appendix \ref{proof of Theorem3}.
\end{proof}

We conclude that the optimal values $q_t^{\text{o}}$ and $q_d^{\text{o}}$ can be found by any convex optimization tool.

\begin{cor}\label{cor2}
When the power budget of the jammer is high (i.e., $\mathcal{P}_w \gg \mathcal{P} \geq 1$), the solution for $\mathfrak{J}$ is given by
\begin{equation}\label{opt_sol_jam}
\begin{dcases*}
\begin{aligned}
&q_t^{\text{o}}=\frac{\mathcal{P}_wT}{2\tau}, \\
&q_d^{\text{o}}=\frac{\mathcal{P}_wT}{2\left(T-\tau \right)}.
\end{aligned}
\end{dcases*}
\end{equation}
\end{cor}

\begin{proof}
Please refer to Appendix \ref{proof of Corollary2}.
\end{proof}
The jammer assigns equal powers to both the pilot and data phases to minimize the SE of the target user, when the jamming power budget is high.

\subsection{Computational Complexity}\label{Compelxity}
The optimization problem $\mathfrak{L}$ relies on the channel statistics, thus it only needs to be solved once for a given setup. To find the optimal values numerically for one setup, we denote the ratio of the power budget that each user allocates to the pilot phase by $\phi \in \left(0,1 \right)$, thus we have $p_t=\phi \mathcal{P}T/\tau$ and $p_d=(1-\phi) \mathcal{P}T/(T-\tau)$ which satisfy the constraints in $\mathfrak{L}$. We should compute $\rho_m^{\star}$ just for a few values of the variable $\phi$ (e.g., $99$ values with equal step size $0.01$) in the interval $\left(0,1 \right)$, and then compare the results to find their maximum value. For each value of $\phi$, we compute the inversion of a matrix with dimensions $M^2 \times M^2$  to obtain $\rho_{m}^{\star}$ according to \eqref{opt_sjnr_r2}, subsequently the complexity is of the order $M^6$. To reduce the complexity, we can rewrite \eqref{opt_sjnr_r2} as

\begin{equation}\label{opt_sjnr_rr22}
\rho_m^{\star} = \tau p_t p_d \Bigg(\tr \left(\matr{R}_m^H\matr{Q}^{-1}\matr{R}_m\matr{B}_{m}^{-1} \right) - 
\frac{\tr \left(\matr{R}_m^H\matr{Q}^{-1}\matr{S}_w^{\left(t\right)}\matr{B}_{m}^{-1} \right) \times \tr \left(\matr{S}_w^{\left(d\right)^H}\matr{Q}^{-1}\matr{R}_m\matr{B}_{m}^{-1} \right)}{1+\tr \left(\matr{S}_w^{\left(d\right)^H}\matr{Q}^{-1}\matr{S}_w^{\left(t\right)}\matr{B}_{m}^{-1} \right)} \Bigg) ,
\end{equation}
following the matrix inversion lemma \cite{CovModEmil} and the equality \eqref{vectorized_property} in Appendix. Utilizing \eqref{opt_sjnr_rr22}, the inversion is computed for the matrices with dimensions $M \times M$ and consequently the order of the complexity is reduced to $M^3$. Following a similar approach, for the problem $\mathfrak{J}$ in \eqref{optimization user1}, the complexity is obtained of the order $M^3$.

\section{Jamming-Robust Massive MIMO Uplink With Hardware Impairments}
The previous sections considered robustness towards jamming for a massive MIMO system with ideal transceiver hardware. However, practical transceivers are equipped with non-ideal hardware that creates distortion in the system. In this section, we investigate how the distortions from hardware impairments at the users and the jammer affect the proposed method for combating jamming.

\subsection{Pilot Phase Under Hardware Impairments}
The distortion caused by the hardware impairments at a user/jammer is modeled by a reduction of the signal amplitude related to the impairment level and the additive Gaussian distortion with a power that equals to the reduction in the signal powers \cite{SecureCellFreeZhang}, \cite{EmilNonIdealHard2018} as

\begin{equation}\label{received_matrix_HarImpair}
\matr{Y\!}_t = \sum_{i=1}^{K}\matr{h}_{i} \left(\sqrt{\tau p_t \left(1-\kappa_u^2 \right)} \boldsymbol{\phi}^T_{i} + \boldsymbol{\eta}^T_{i} \right) \\
+ \matr{h}_w \left(\sqrt{\tau q_t \left(1-\kappa_w^2 \right)} \boldsymbol{\phi}^T_{m} + \boldsymbol{\eta}^T_{w} \right) + \matr{N}_t ,
\end{equation}
where $\kappa_u$ and $\kappa_w$ represent the impairment levels at the users and the jammer respectively that are in the ranges of $0 \leq \kappa_u \leq 1$ and $0 \leq \kappa_w \leq 1$. It is assumed that the distortion terms $\boldsymbol{\eta}_i$ and $\boldsymbol{\eta}_w$ are independently distributed as

\begin{equation}\label{DistNoiseDist}
\boldsymbol{\eta}_i \sim \mathcal{CN}\left(\matr{0}, p_t \kappa_u^2 \matr{I}_{\tau} \right) ~ \text{and} ~ \boldsymbol{\eta}_w \sim \mathcal{CN}\left(\matr{0}, q_t \kappa_w^2 \matr{I}_{\tau} \right) .
\end{equation}
In the system with non-ideal hardware, the channel estimation is more complicated due to the distortion terms. A basic estimate of $\matr{h}_k$ can be obtained from

\begin{equation}\label{proj_est_hardware_imp}
\matr{y}_{k} = \matr{Y\!}_t \boldsymbol{\phi}_{k}^* = \\
\begin{dcases*}
\begin{aligned}
&\sqrt{\tau p_t \left(1-\kappa_u^2 \right)} \matr{h}_{m} + \sum_{i=1}^{K} \varepsilon_{i}^{m} \matr{h}_{i} + \\
&~~~~~~~~~~~\sqrt{\tau q_t \left(1-\kappa_w^2 \right)} \matr{h}_w + \varepsilon_{w}^{m} \matr{h}_{w} + \matr{n}_{m} , ~ && k=m , \\
&\sqrt{\tau p_t \left(1-\kappa_u^2 \right)} \matr{h}_{k} + \sum_{i=1}^{K} \varepsilon_{i}^{k} \matr{h}_{i} + \varepsilon_{w}^{k} \matr{h}_{w} + \matr{n}_{k} , ~ && k \neq m ,
\end{aligned}
\end{dcases*}
\end{equation}
where $\varepsilon_{i}^{k} \triangleq \boldsymbol{\eta}^T_{i} \boldsymbol{\phi}_{k}^* \sim \mathcal{CN} (0, p_t \kappa_u^2)$ and $\varepsilon_{w}^{k} \triangleq \boldsymbol{\eta}^T_{w} \boldsymbol{\phi}_{k}^* \sim \mathcal{CN} (0, q_t \kappa_w^2)$ that are independent of the channels. We exploit $\hat{\matr{h}}_{k} = \matr{A}_k \matr{y}_k$ as a linear estimate of the $k$th user's channel. The covariance matrix of $\matr{y}_k$ in \eqref{proj_est_hardware_imp} is given by

\begin{equation}\label{cov_est_hardware_imp}
\tilde{\matr{B}}_k = \text{E} \left\{\matr{y}_k \matr{y}_k^H \right\} = \\
\begin{dcases*}
\begin{aligned}
&\tau p_t \left(1-\kappa_u^2 \right) \matr{R}_m + \\ &~~~~~~~~ p_t \kappa_u^2 \sum_{i=1}^{K} \matr{R}_i + \tilde{\matr{S}}^{\left(t\right)}_w + \matr{I}_M , ~ && k=m , \\
&\tau p_t \left(1-\kappa_u^2 \right) \matr{R}_k + \\ &~~~~~~~~ p_t \kappa_u^2 \sum_{i=1}^{K} \matr{R}_i + \breve{\matr{S}}^{\left(t\right)}_w + \matr{I}_M , ~ && k \neq m ,
\end{aligned}
\end{dcases*}
\end{equation}
where $\tilde{\matr{S}}^{\left(t\right)}_w \triangleq q_t \left(\tau -\left(\tau -1\right) \kappa_w^2 \right) \matr{R}_w$ and $\breve{\matr{S}}^{\left(t\right)}_w \triangleq q_t \kappa_w^2 \matr{R}_w$.

\subsection{Robust Data Phase Under Hardware Impairments}
The received signal at the BS during the data phase for the case of hardware impairments at the users and the jammer is modeled by

\begin{equation}\label{rec_sig_HardImpair}
\matr{y}_d = \sum_{i=1}^{K} \matr{h}_{i} \left(\sqrt{p_d \left(1-\kappa_u^2 \right)} x_{i} + \eta_i \right) \\
+ \matr{h}_w \left(\sqrt{q_d \left(1-\kappa_w^2 \right)} x_{w} + \eta_w \right) + \matr{n}_d .
\end{equation}
\begin{lem} \label{Lemma2}
Under hardware impairments, an achievable SE of the $k$th user in \eqref{erg_cap_imperfect} becomes

\begin{equation}
\label{erg_cap_imperfect_HardImpair}
\mathcal{S}_k = \left(1-\frac{\tau}{T}\right) \log_2 \left(1+\tilde{\rho}_{k} \right) ,
\end{equation}
where the effective SINR, $\tilde{\rho}_k$, is given by \eqref{sjnr_imperfect_HardImpair}, shown at the top of the previous page.
\end{lem}

\begin{proof}
The proof is similar to the proof of Lemma \ref{Lemma1}, except that the numerator and second term in the denominator of the SINR in \eqref{eff_sjnr_kth_user} are multiplied by $(1-\kappa_u^2)$.
\end{proof}
While the hardware impairments at the users reduce the gain in the numerator of the SINR in \eqref{sjnr_imperfect_HardImpair} by a factor $(1-\kappa_u^2)^2$ and also unwanted pilot-contaminating interferers (i.e., the second and third terms in the denominator), the hardware impairment at the jammer also creates pilot-contaminating interference that imposes additional reduction in the SINR. The reduction of the SINR through the non-ideal hardware of the jammer is more for the $m$th user in comparison with the other users, since $\tau -\left(\tau -1\right)\kappa_w^2 \geq \kappa_w^2$ due to the bound of $\kappa_w \leq 1$.

To derive the MS estimator in the case of non-ideal hardware, \eqref{sjnr_imperfect_HardImpair} can first be rewritten as

\setcounter{equation}{34}
\begin{equation}\label{sjnr_imperfect_HardImpair2}
\tilde{\rho}_{k} = 
\frac{\tau p_tp_d\left(1-\kappa_u^2 \right)^2 \matr{a}_k^H \matr{r}_k \matr{r}_k^H \matr{a}_k}{\matr{a}_k^H \tilde{\matr{C}}_k \matr{a}_k} ,
\end{equation}
where $\tilde{\matr{C}}_k$ is the $M^2 \times M^2$ matrix defined by

\begin{equation}\label{Matrix_C_HardImpair}
\tilde{\matr{C}}_k \triangleq \\
\begin{dcases*}
\begin{aligned}
&\tilde{\matr{B}}_{m}^T \otimes \matr{Q} + \tilde{\matr{s}}_w^{\left(t\right)} \matr{s}_w^{\left(d\right)^H} + p_tp_d \kappa_u^2 \sum_{\substack{i=1}}^{K} \matr{r}_i \matr{r}_i^H + \\
&~~~~~~~~~~~~~~~~\tau p_tp_d\kappa_u^2 \left(1-\kappa_u^2 \right) \matr{r}_m \matr{r}_m^H , && k=m , \\
&\tilde{\matr{B}}_{k}^T \otimes \matr{Q} + \breve{\matr{s}}_w^{\left(t\right)} \matr{s}_w^{\left(d\right)^H} + p_tp_d \kappa_u^2 \sum_{\substack{i=1}}^{K} \matr{r}_i \matr{r}_i^H + \\
&~~~~~~~~~~~~~~~~~\tau p_tp_d\kappa_u^2 \left(1-\kappa_u^2 \right) \matr{r}_k \matr{r}_k^H , && k \neq m ,
\end{aligned}
\end{dcases*}
\end{equation}
where $\tilde{\matr{s}}_w^{\left(t\right)} = \vect (\tilde{\matr{S}}_w^{\left(t\right)}) \in \mathbb{C}^{M^2 \times 1}$ and $\breve{\matr{s}}_w^{\left(t\right)} = \vect (\breve{\matr{S}}_w^{\left(t\right)}) \in \mathbb{C}^{M^2 \times 1}$. As a result, the MS estimator can be readily extended as $\matr{a}_k^{\star} = \tilde{\matr{C}}_k^{-1} \matr{r}_k$ in this case. Consequently, the optimal effective SINR is computed as $\tilde{\rho}_k^{\star} = \tau p_tp_d(1-\kappa_u^2)^2 \matr{r}_k^H \tilde{\matr{C}}_k^{-1} \matr{r}_k$. We can apply the MS approach using the matrix $\tilde{\matr{C}}_k$ in \eqref{Matrix_C_HardImpair} and during the sample estimate of the matrices $\tilde{\matr{S}}_w^{\left(t\right)}$ and $\breve{\matr{S}}_w^{\left(t\right)}$, the observations of the pilot signals $\matr{y}_k = \matr{Y}_t \boldsymbol{\phi}_k^*$ are accomplished from $\matr{Y}_t$ in \eqref{received_matrix_HarImpair} over $N$ coherence blocks.

Finally, in the case of non-ideal hardware, the optimal power allocation between the pilot and data phases from the system's and the jammer's point of view is analyzed similar to the case of ideal hardware.

\begin{figure*}[!t]
\setcounter{equation}{38}
\begin{equation}\label{eff_sjnr_kth_user}
\rho_k = \frac{p_d \left|\text{E} \left\{\hat{\matr{h}}_k^H \matr{h}_{k} \right\} \right|^2}{p_d \sum_{i=1}^{K} \text{E} \left\{\left|\hat{\matr{h}}_k^H \matr{h}_{i} \right|^2 \right\} - p_d \left|\text{E} \left\{\hat{\matr{h}}_k^H \matr{h}_{k} \right\} \right|^2 + q_d \text{E} \left\{\left|\hat{\matr{h}}_k^H \matr{h}_{w} \right|^2 \right\} + \text{E} \left\{\left\|\hat{\matr{h}}_k \right\|^2 \right\}} 
\end{equation}
\hrulefill
\end{figure*}

\begin{figure*}
\centerline{\subfigure[User 1]{\includegraphics[width=2.35in]{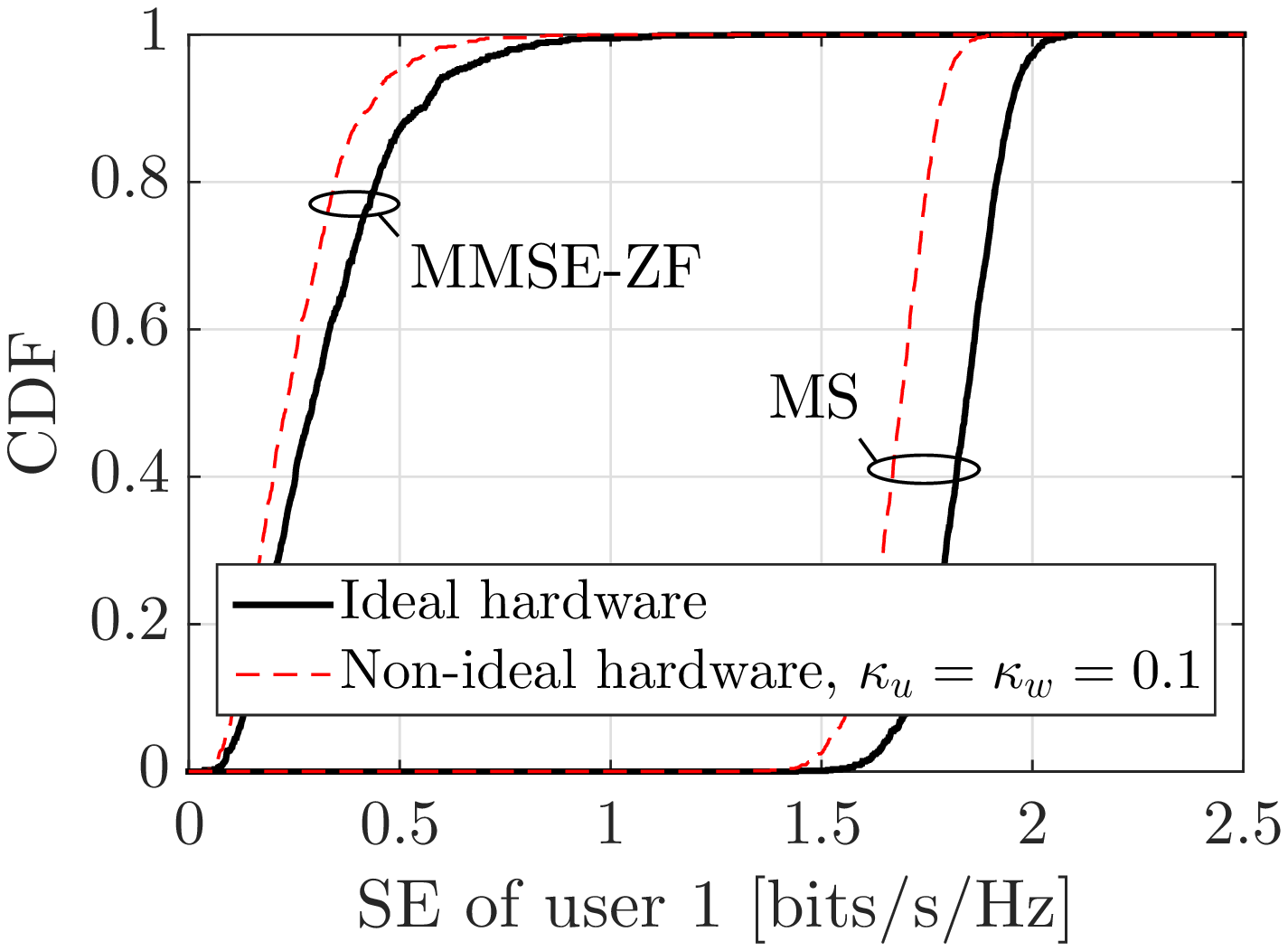}
\label{fig0_1}}
\hfil
\subfigure[User 2]{\includegraphics[width=2.35in]{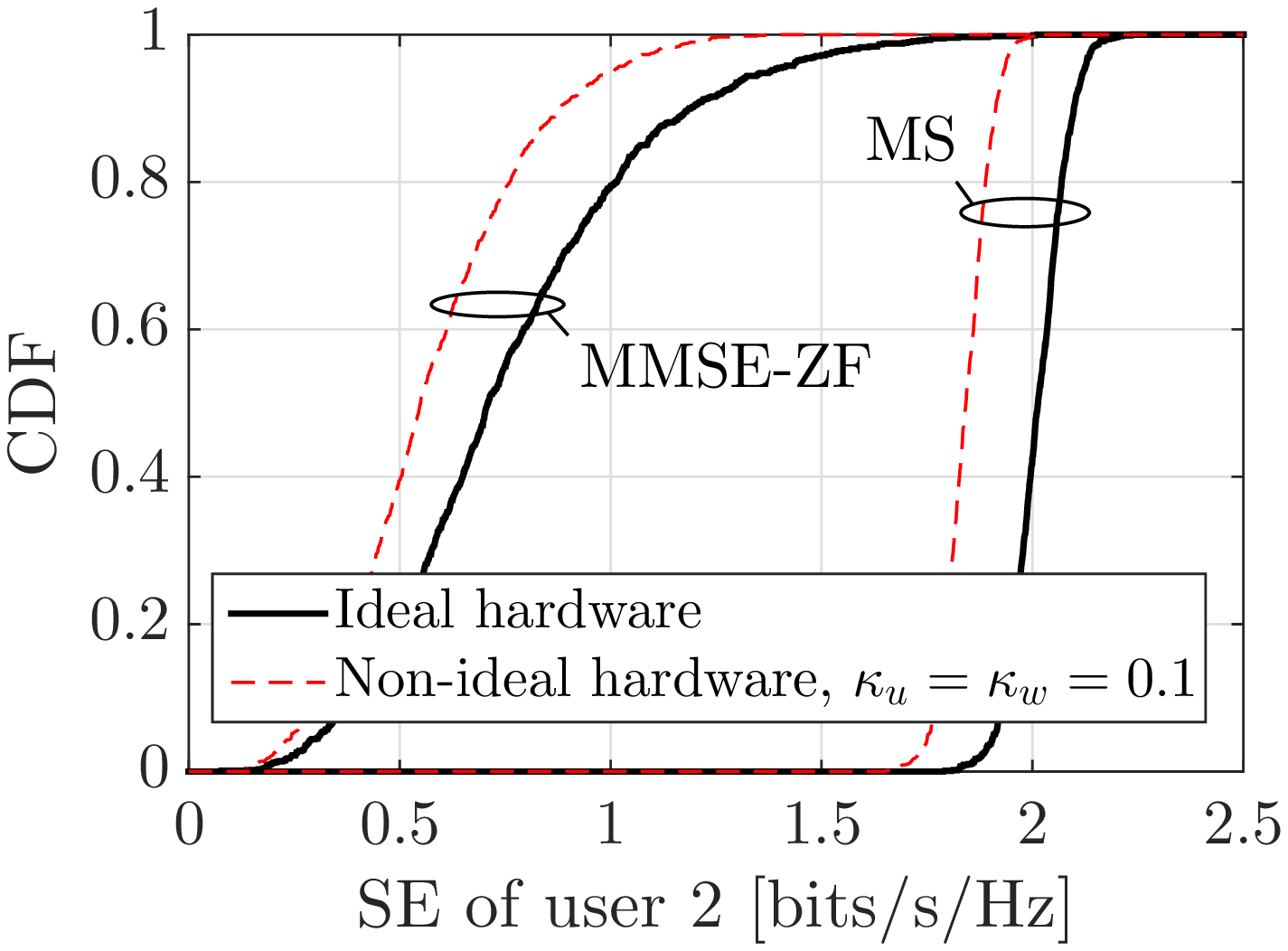}
\label{fig0_2}}
\hfil
\subfigure[User 3 (attacked user)]{\includegraphics[width=2.35in]{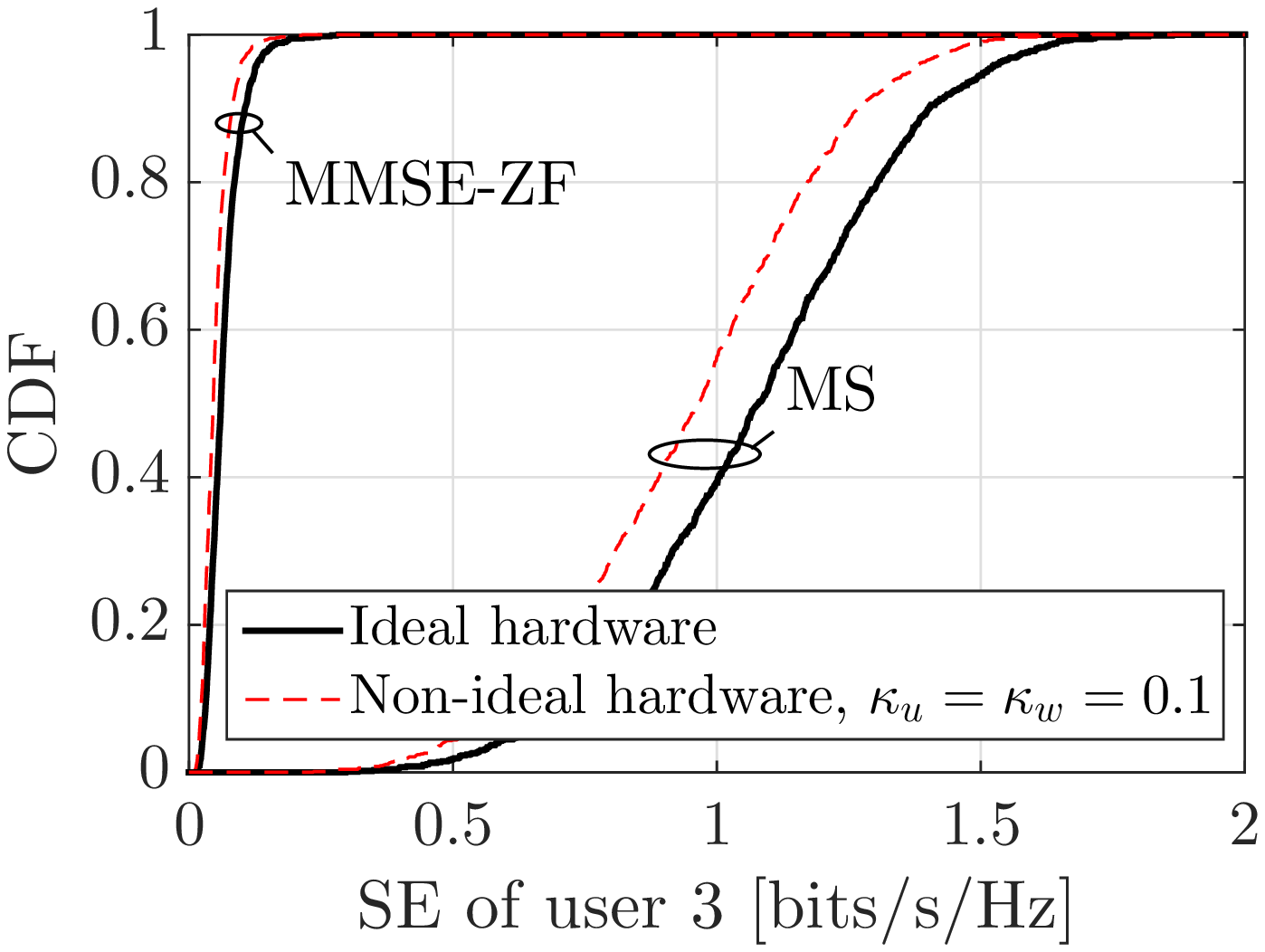}
\label{fig0_3}}}
\caption{CDF of SE for each user computed on $1000$ different random locations of the jammer and fading variations over the array.}
\label{fig0}
\end{figure*}

\section{Numerical Experiments}
We consider a square cell of $250~\m \times 250~\m$ composed of one BS located at the cell center and $K$ users that are randomly distributed in the cell with a minimum distance of $25~\m$ from the BS. We assume a coherence block of $T=200$ samples with a pilot length $\tau=K$. For uncorrelated Rayleigh fading channels, it was shown in the literature \cite{JamResTai}, \cite{JamSupHoss} that the massive MIMO uplink is robust against jamming attacks with a MMSE-like estimator and a ZF-like detector. Hence, here we compare the performance of the MS approach with a framework including a MMSE estimator and a ZF detector (with label ``MMSE-ZF'') as a benchmark. To compute the SE for the MMSE-ZF, we utilize Monte-Carlo simulations over at least $1000$ independent runs. In the simulations, we include the noise variance in what we call the large-scale fading coefficients, not in the transmit powers, thus the transmit powers can be interpreted as the signal-to-noise ratios and their units are expressed in $\dB$. In the scenario with non-ideal hardware, the impairment levels for the users and the jammer are set to $\kappa_u=\kappa_w=0.1$.

For all the users and the jammer, the covariance matrices are computed as a combination of the Gaussian local scattering model and log-normal large-scale fading variations over the array \cite[Ch. 2]{EmilBook}:

\setcounter{equation}{36}
\begin{equation}\label{cov_mat_model}
\left[\matr{R}_l \right]_{mn} = \frac{\beta_l 10^{\left(f^l_{m} + f^l_{n} \right)/20}}{\sqrt{2 \pi}\sigma_{\varphi}} \int e^{j\pi \left(m-1 \right) \left(n-m\right) \sin \left(\theta_l + \varphi \right)-\frac{\varphi}{2\sigma^2_{\varphi}}} \ddd \varphi , \\
l \in \left\{1, 2, \ldots, K \right\} \cup \left\{w \right\} .
\end{equation}
The parameters in \eqref{cov_mat_model} are as follows:
\begin{align}
\beta_l & \text{: Average large-scale fading (depends on the distance} \nonumber \\[-0.45em]
& \text{~~from the BS and includes the noise variance).} \nonumber \\
\theta_l & \text{: Angle-of-arrival (AoA).} \nonumber \\
f^l_{m}, f^l_{n} & \text{: Independent large-scale fading variations distributed} \nonumber \\[-0.45em]
& \text{~~as $\mathcal{N} \left(0, \sigma^2 \right)$. The standard deviation is denoted by $\sigma$.} \nonumber\\[-0.45em]
\sigma_{\varphi} & \text{: Angular standard deviation ($\sigma_{\varphi} = 10^{\circ}$).} \nonumber
\end{align}

To perform the simulations of Fig. \ref{fig0}, Fig. \ref{fig4}, and Fig. \ref{fig6}, we use a realization of the introduced network containing a $64$-antenna BS and $3$ users with locations $\left(-57.6~\m,112.2~\m\right)$, $\left(44.4~\m,-83.1~\m\right)$, and $\left(103.8~\m,-22.2~\m\right)$. One jammer also exists in the network that attacks the $3^{\text{rd}}$ user as a target. We set the standard deviation of the large-scale fading variations to $\sigma=4$.

Fig. \ref{fig0} illustrates the cumulative distribution function (CDF) of the SE for each user computed on $1000$ different random locations of the jammer and fading variations over the array. The curves are plotted for $p_t=p_d=-5 \dB$ and $q_t=q_d=15 \dB$ in both scenarios with ideal and non-ideal hardware. The comparison shows that the MS approach generally performs better than the MMSE-ZF specially in the case of the non-ideal hardware at the transmitters. Moreover, the robustness level of the system during the detection of the $3^{\text{rd}}$ user's symbols (i.e., the target user's symbols) under the jamming attacks is dominantly improved using the MS approach compared to the MMSE-ZF.

\begin{figure*}
\centerline{\subfigure[$q_t=q_d=15 \dB$ and $\mathcal{P}=5 \dB$]{\includegraphics[width=2.55in]{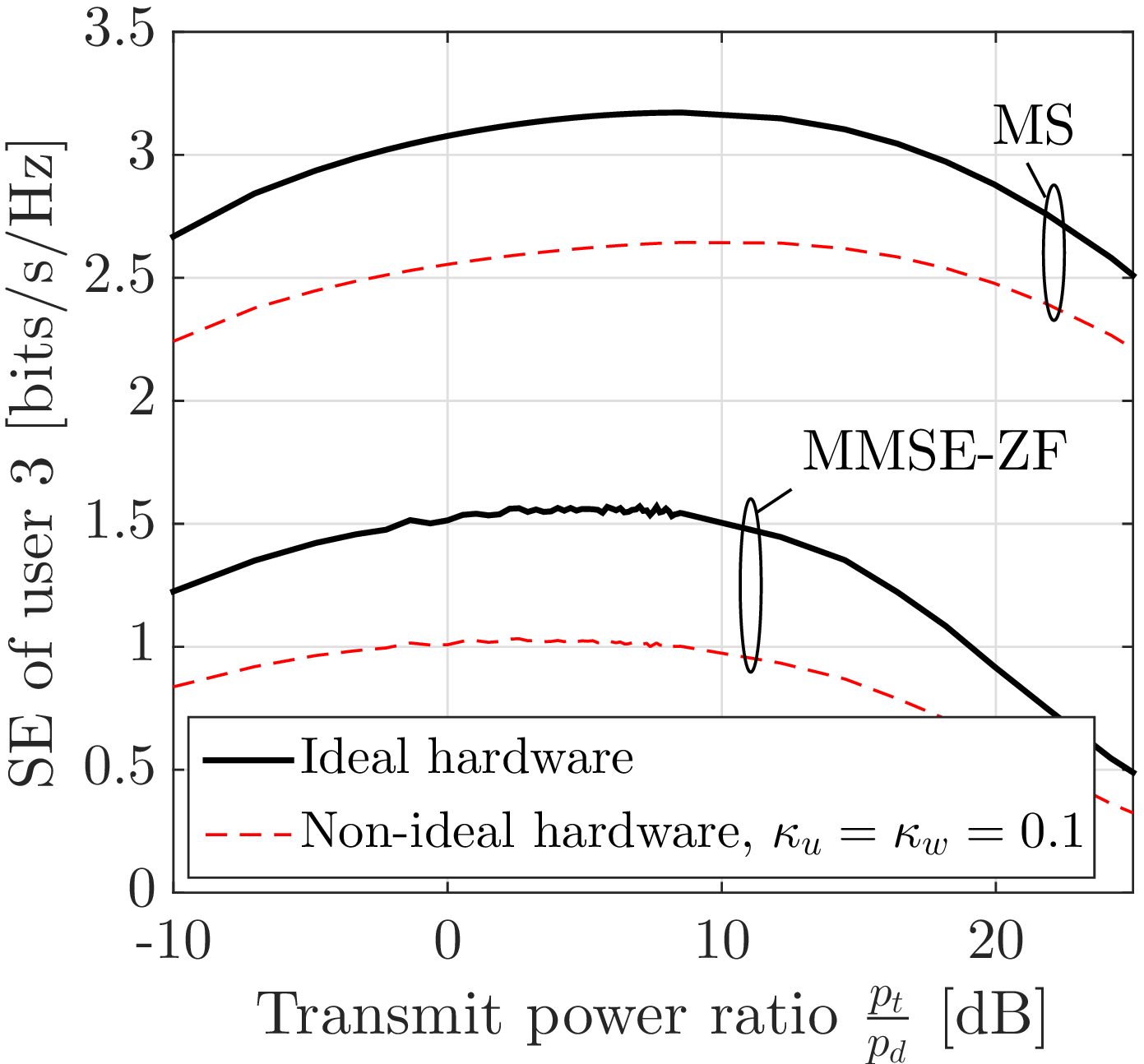}
\label{fig4_1}}
\hfil
\subfigure[$p_t=p_d=-5 \dB$ and $\mathcal{P}_w=20 \dB$]{\includegraphics[width=2.55in]{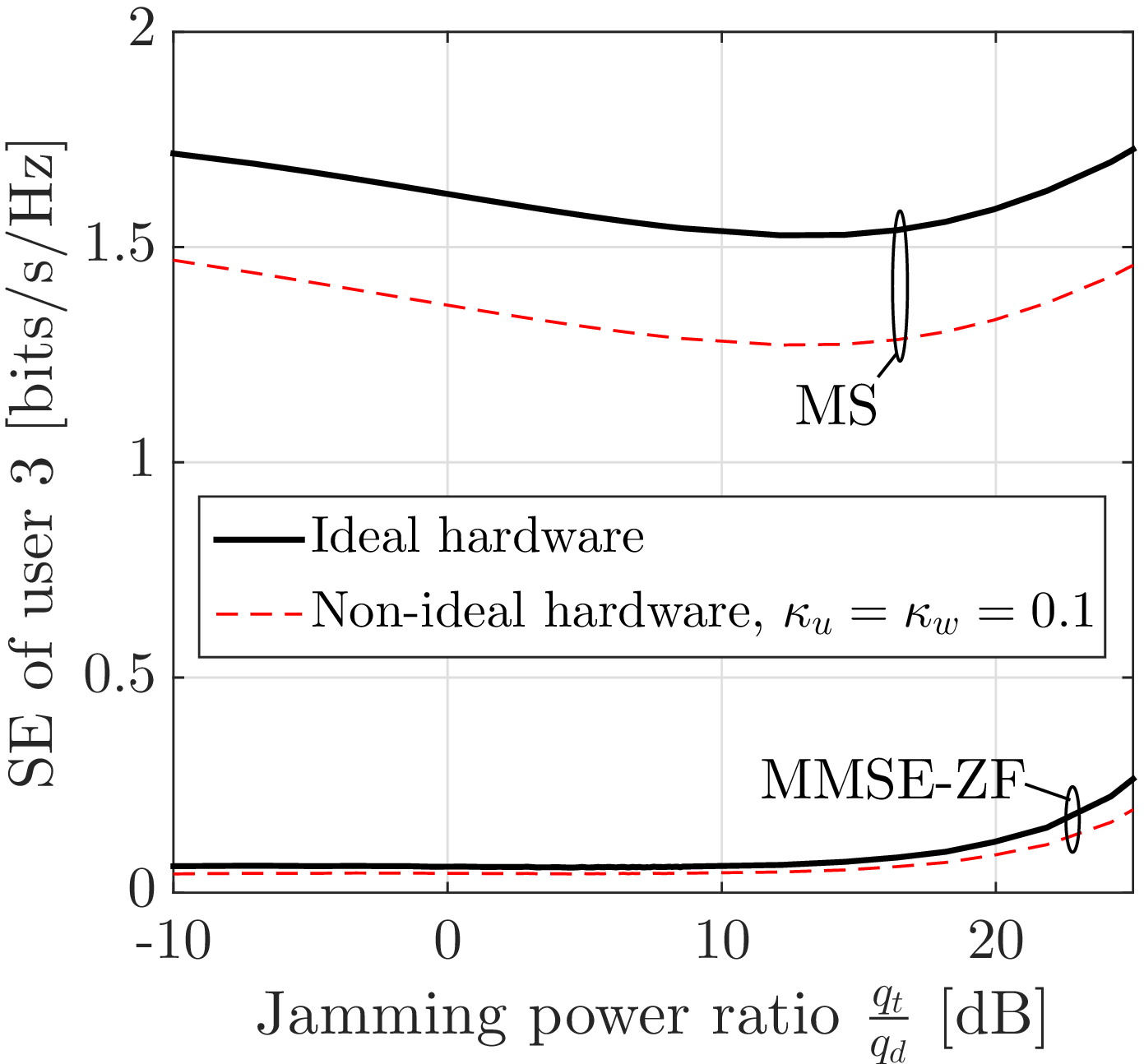}\label{fig4_2}}}
\caption{SE of user $3$ (target user) versus (a) transmit power ratio $\frac{p_t}{p_d}$, (b) jamming power ratio $\frac{q_t}{q_d}$.}
\label{fig4}
\end{figure*}

We evaluate the impact of power allocations on the performance of the system in Fig. \ref{fig4}. In these simulations, the location of the jammer is near to the target user in order to create a channel correlation matrix $\matr{R}_w$ similar to the channel correlation matrix $\matr{R}_3$. It intuitively seems that the jammer can more easily spread its jamming signals on the subspace of the target user by choosing a location near to the target user. For instance, here we assume that the jammer's location is $\left(107.9~\m,-17.1~\m\right)$. Fig. \ref{fig4_1} illustrates how the SE of the target user depends on the transmit power ratio $\frac{p_t}{p_d}$ for $q_t=q_d=15 \dB$ and $\mathcal{P}=5 \dB$ in the scenarios with the ideal and non-ideal hardware, while in Fig. \ref{fig4_2}, the SE of the target user is depicted versus the jamming power ratio $\frac{q_t}{q_d}$ for $p_t=p_d=-5 \dB$ and $\mathcal{P}_w=20 \dB$ at the same scenarios. we can see that in Fig. \ref{fig4_1}, the SEs have the peak values in certain power ratios $\frac{p_t}{p_d}$. In other words, the users should balance their power budgets between the pilot and data signals to get the best achievable SEs. Similarly in Fig. \ref{fig4_2}, the SEs have the minimum values in certain ratios $\frac{q_t}{q_d}$. Therefore, the jammer can divide its power budget between the pilot and data phases to minimize the achievable SE of the target user. Furthermore, comparing Fig. \ref{fig4_1} and Fig. \ref{fig4_2} indicates that the performance of the MMSE-ZF dramatically reduces when the jammer's power budget increases compared to the users' power budget, such that its SE in Fig. \ref{fig4_2} varies near the zero; whereas the MS approach is more robust against jamming attacks even with strong jamming powers.

Finally, we evaluate the accuracies of the sample estimations for $\hat{\matr{S}}_w^{\left(t\right)}$ in \eqref{EstVar_s_w} and $\hat{\matr{S}}_w^{\left(d\right)}$ in \eqref{Est_Rec_Jam_Pow1}. We consider $p_t=p_d=-10 \dB$ and the location of the jammer is identical to that location in the simulations of Fig. \ref{fig4}. We plot the SE of the MS approach for all $3$ users versus the jamming powers $q_t=q_d$ in Fig. \ref{fig6} for ideal hardware and two cases that are perfect knowledge about $\matr{S}_w^{\left(t\right)}$ and $\matr{S}_w^{\left(d\right)}$, and the sample estimates of $\hat{\matr{S}}_w^{\left(t\right)}$ and $\hat{\matr{S}}_w^{\left(d\right)}$ over $N=1000$ coherence blocks. It demonstrates that the SEs resulted from the estimations in \eqref{EstVar_s_w} and \eqref{Est_Rec_Jam_Pow1} nearly behave the actual SEs for all $3$ users. Accordingly, it is expected that the achievable SEs can be realized without the perfect knowledge of the jamming channel statistics. Furthermore, we can see in Fig. \ref{fig6} that the SEs become non-zero in the high jamming power regime. This is because we can estimate the jammer's interference better for a high power jammer, and therefore remove its impact spatially by exploiting our proposed robust framework.

Another realization of the introduced network with $5$ users are utilized to perform the next simulation to show the effect of the number of the BS antennas on the performance of the MS approach. The locations of the users are $\left(-59.1~\m,-91.6~\m\right)$, $\left(32.6~\m,-113.3~\m\right)$, $\left(25.1~\m,116.9~\m\right)$, $\left(94.0~\m,106.3~\m\right)$, and $\left(33.2~\m,1.6~\m\right)$. One jammer located in $\left(37.0~\m,6.2~\m\right)$ attacks the $5^{\text{th}}$ user as a target. We use \eqref{cov_mat_model} to carry out the simulation of the correlation matrices except that here we ignore the fading variations over the array, i.e., $f_m^l=f_n^l=0$. In Fig. \ref{fig8}, the SEs of the users $1$, $3$, and $5$ are illustrated with respect to the number of BS antennas for $\mathcal{P}=5~\dB$ and $\mathcal{P}_w=25~\dB$. Since the SEs of the users $1$, $2$, $3$, and $4$ (i.e., the users that are not the target of the jammer) are similar, we indicate the SEs of the users $1$, and $3$ as examples. The users and the jammer transmit the signals in the pilot and data phases with the transmit powers allocated by \eqref{opt_sol_leg_sys} and \eqref{opt_sol_jam}, respectively. It is clear that the performance of the proposed framework improves for the large antenna arrays. This improvement is substantial in the target user (i.e., the $5^{\text{th}}$ user) such that the SE of the MS approach grows by increasing the number of antennas, while for the MMSE-ZF, the SE remains nearly zero.

\begin{figure}
\centering
\includegraphics[width=2.55in]
{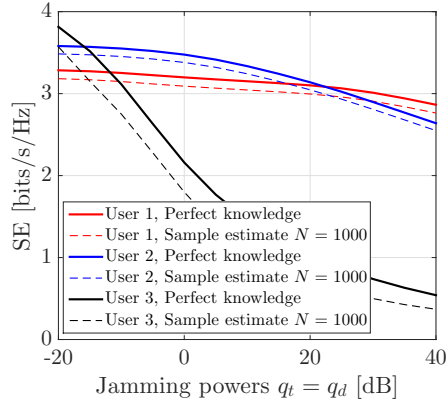}
\caption{SEs of the MS approach versus the jamming powers $q_t=q_d$ for the ideal hardware and $p_t=p_d=-10 \dB$.}
\label{fig6}
\end{figure}

\begin{figure*}
\centerline{\subfigure[User 1]{\includegraphics[width=2.35in]{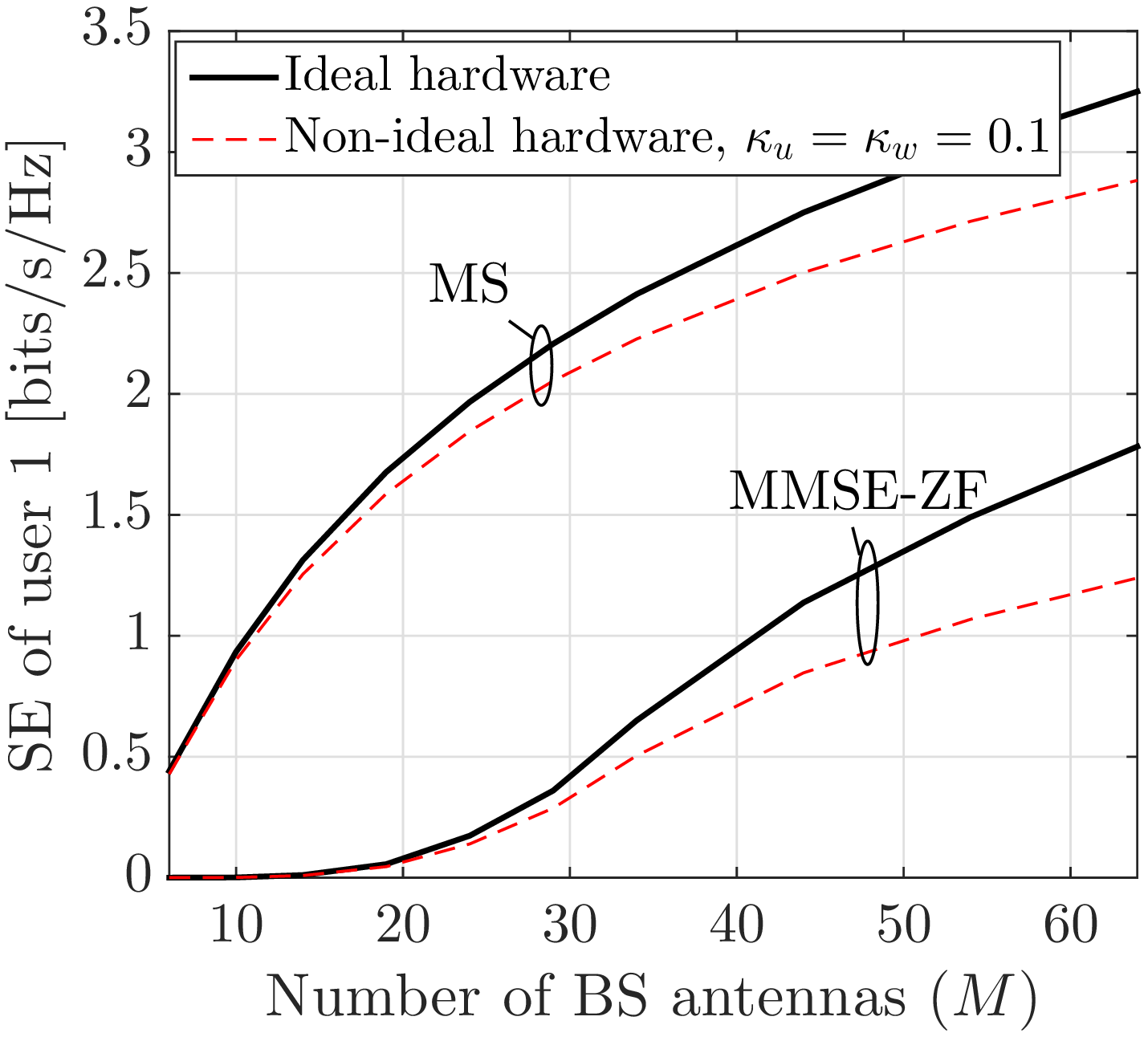}
\label{fig8_1}}
\hfil
\subfigure[User 3]{\includegraphics[width=2.35in]{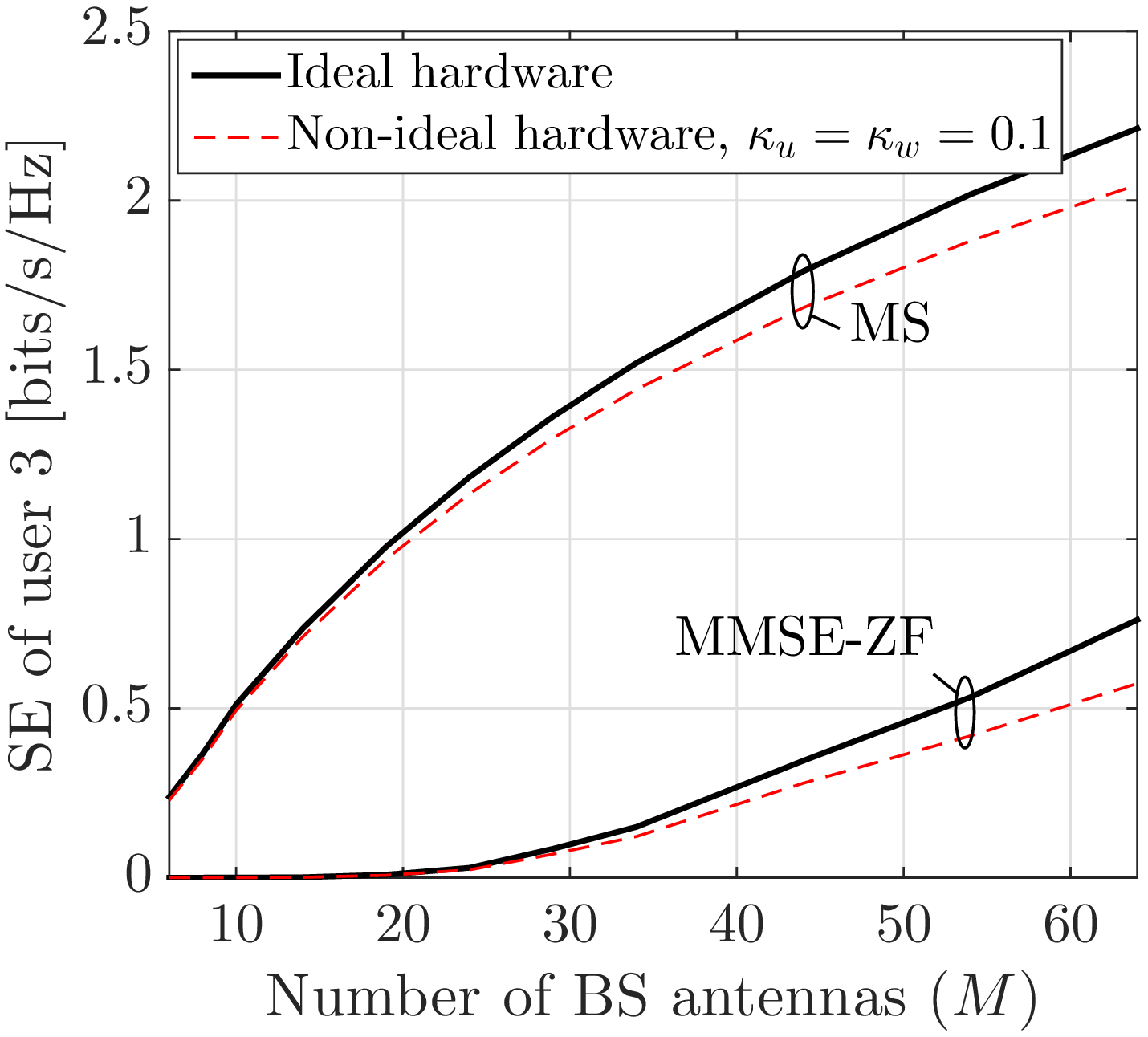}
\label{fig8_3}}
\hfil
\subfigure[User 5 (attacked user)]{\includegraphics[width=2.35in]{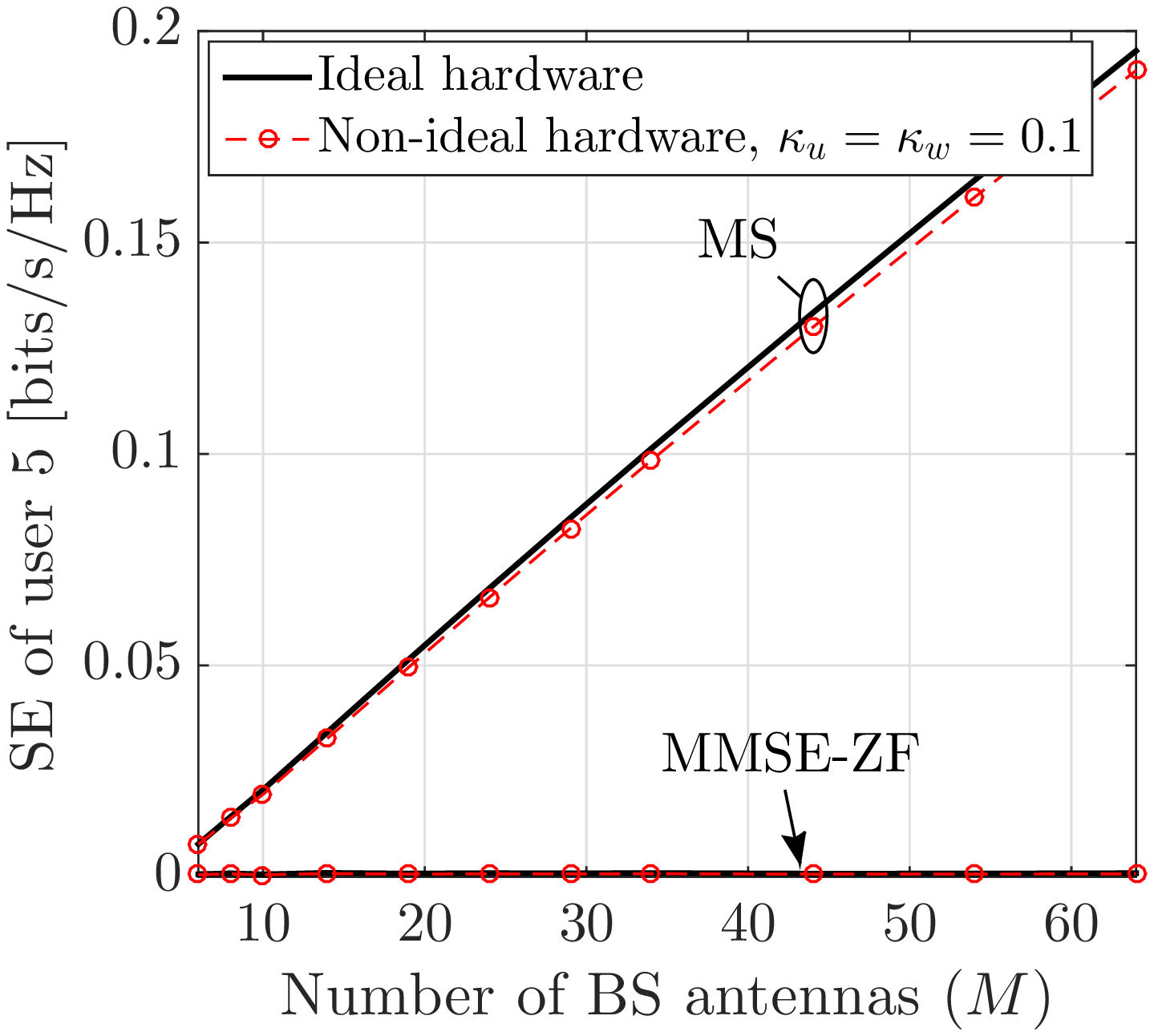}
\label{fig8_5}}}
\caption{SEs of the users $1$, $3$, and $5$ with respect to the number of the BS antennas $M$ for $\mathcal{P}=5 \dB$ and $\mathcal{P}_w=25 \dB$.}
\label{fig8}
\end{figure*}

\section{Conclusion}
This paper has investigated the protection of the massive MIMO uplink with the spatially correlated Rayleigh fading against jamming attacks. A new jamming-robust framework has been proposed consisting of a new MS approach as an optimal channel estimator with respect to the SE of the system. A practical algorithm has been also suggested to acquire the jamming statistical knowledge needed for implementation of the proposed framework. We have shown that a massive MIMO system equipped with the proposed framework is more robust against jamming compared to the MMSE-ZF framework. We have also shown that the power allocations at the users can improve the performance of the proposed framework regardless of the jamming power optimization. Finally, the performance of the proposed framework has been mathematically and numerically studied when the hardware impairments exist at the users and the jammer. The analyses show that the proposed framework performs well in the case of non-ideal hardware for both the attacked user and other users that are not the target of the jammer.

\appendix
\subsection{Proof of Lemma \ref{Lemma1}} \label{proof of Lemma1}
An achievable SE of the $k$th user is
\setcounter{equation}{37}
\begin{equation}\label{ach_SE_kth_user}
\mathcal{S}_k = \left(1-\frac{\tau}{T} \right) \log_2 \left(1 + \rho_k \right) ,
\end{equation}
with the effective SINR of the $k$th user, $\rho_k$, is obtained as in \eqref{eff_sjnr_kth_user}, shown at the top of the previous page. The expectations in \eqref{eff_sjnr_kth_user} are computed in closed form as follows:
\setcounter{equation}{39}
\begin{equation}\label{Nominator}
\E \left\{\hat{\matr{h}}_k^H \matr{h}_{k} \right\} = \sqrt{\tau p_t} \tr \left(\matr{A}_k^H \matr{R}_k \right) ,
\end{equation}
\begin{equation}\label{second_order_moment_users}
\E \left\{\left|\hat{\matr{h}}_{k}^H \matr{h}_{i} \right|^2 \right\} =
\begin{dcases*}
\begin{aligned}
&\tau p_t \left|\tr \left(\matr{A}_k^H \matr{R}_k \right)\right|^2 + \\
&~~~~~~~~~\tr \left( \matr{A}_k^H \matr{R}_{k} \matr{A}_k \matr{B}_k \right) , ~ && k = i , \\
&\tr \left( \matr{A}_k^H \matr{R}_{i} \matr{A}_k \matr{B}_k \right) , ~ && k \neq i ,
\end{aligned}
\end{dcases*}
\end{equation}
\begin{equation}\label{second_order_moment_jam2}
\E \left\{\left|\hat{\matr{h}}_{k}^H \matr{h}_{w} \right|^2 \right\} =
\begin{dcases*}
\begin{aligned}
&\tau q_t \left|\tr \left(\matr{A}_m^H \matr{R}_w \right)\right|^2 + \\
&~~~~~~~~~\tr \left( \matr{A}_m^H \matr{R}_{w} \matr{A}_m \matr{B}_m \right) , ~ && k=m , \\
&\tr \left( \matr{A}_k^H \matr{R}_{w} \matr{A}_k \matr{B}_k \right) , ~ && k \neq m ,
\end{aligned}
\end{dcases*}
\end{equation}
\begin{equation}\label{rec_noise}
\E \left\{\left\|\hat{\matr{h}}_k\right\|^2 \right\} = \tr \left( \matr{A}_k^H \matr{A}_k \matr{B}_k \right) .
\end{equation}
By plugging \eqref{Nominator}$-$\eqref{rec_noise} into \eqref{eff_sjnr_kth_user}, and utilizing $\matr{Q} = \matr{I}_M + \matr{S}_w^{\left(d\right)} + p_d \sum_{i=1}^{K} \matr{R}_i$, we finally obtain \eqref{erg_cap_imperfect}.

\subsection{Proof of Theorem \ref{theo1}} \label{proof of Theorem1}
According to $\matr{a}_m \triangleq \vect (\matr{A}_m) \in \mathbb{C}^{M^2 \times 1}$ and the vectorized form of the jamming channel statistic $\matr{r}_w = \vect \left(\matr{R}_w\right) \in \mathbb{C}^{M^2 \times 1}$ and $\matr{r}_m$, the effective SINR in \eqref{sjnr_imperfect} is rewritten as
\begin{equation}\label{sjnr_rewrite}
{\rho}_{m} = \frac{\tau p_t p_d \matr{a}_m^H \matr{r}_m \matr{r}_m^H \matr{a}_m}{\matr{a}_m^H \matr{s}_w^{\left(t\right)} \matr{s}_w^{\left(d\right)^H} \matr{a}_m + \matr{a}_m^H \left(\matr{B}_{m}^T \otimes \matr{Q} \right) \matr{a}_m} ,
\end{equation}
for $k=m$, where $\matr{s}_w^{\left(t\right)}$, $\matr{s}_w^{\left(d\right)}$ are the vectorized form of $\matr{S}_w^{\left(t\right)} = \tau q_t \matr{R}_w$ and $\matr{S}_w^{\left(d\right)} = q_d \matr{R}_w$, and we use the Kronecker product to express matrix multiplication as
\begin{align}\label{vectorized_property}
\tr \left(\matr{A}_m^H \matr{Q} \matr{A}_m \matr{B}_{m} \right) &= \matr{a}_m^H \vect \left(\matr{Q} \matr{A}_m \matr{B}_{m} \right) \nonumber \\
&= \matr{a}_m^H \left(\matr{B}_{m}^T \otimes \matr{Q} \right) \matr{a}_m .
\end{align}
By defining the new matrix $\matr{C}_m \triangleq \matr{B}_{m}^T \otimes \matr{Q} + \matr{s}_w^{\left(t\right)} \matr{s}_w^{\left(d\right)^H}$ and making the change of variable $\bar{\matr{a}}_m = \matr{C}_m^{\frac{1}{2}} \matr{a}_m$, we have
\begin{equation}\label{sjnr_rewrite2}
{\rho}_{m} = \frac{\bar{\matr{a}}_m^H \matr{C}_m^{-\frac{1}{2}} \matr{r}_m \matr{r}_m^H \matr{C}_m^{-\frac{1}{2}} \bar{\matr{a}}_m}{\bar{\matr{a}}_m^H \bar{\matr{a}}_m} = \frac{\left|\bar{\matr{a}}_m^H \matr{C}_m^{-\frac{1}{2}} \matr{r}_m \right|^2}{\left\|\bar{\matr{a}}_m \right\|^2} .
\end{equation}
According to Cauchy-Schwarz inequality for the numerator, ${\rho}_{m}$ is maximized when $\bar{\matr{a}}_m^{\star} = \matr{C}_m^{-\frac{1}{2}} \matr{r}_m$, which leads to $\matr{a}_m^{\star} = \matr{C}_m^{-1} \matr{r}_m$.

For $k \neq m$, the analysis is similar to $k=m$, except that the first term in the denominator of the SINR in \eqref{sjnr_rewrite} is ignored.

\subsection{Proof of Corollary \ref{cor1}} \label{proof of Corollary1}
Regarding \eqref{opt_sjnr_r2}, $\rho_m^{\star}$ is approximated to
\begin{equation}\label{opt_sjnr_r2_1}
\rho_m^{\star} = \tau p_t p_d \matr{r}_m^H \left(\matr{S}_{w}^{\left(t\right)^T} \otimes \matr{S}_{w}^{\left(d\right)} + \matr{s}_w^{\left(t\right)} \matr{s}_w^{\left(d\right)^H} \right)^{-1} \matr{r}_m ,
\end{equation}
for $\mathcal{P}_w \gg \mathcal{P} \geq 1$. This approximation is exact as $\mathcal{P}_w \rightarrow \infty$. We can see that in \eqref{opt_sjnr_r2_1}, ${\rho}_m^{\star}$ is a linear function of the multiplication $p_t p_d$. Hence, $\mathfrak{L}$ equals to
\begin{equation}\label{optimization user2}
\mathfrak{L}':
\begin{dcases*}
\begin{aligned}
& \underset{p_t, p_d}{\text{maximize}}
&& p_t p_d \\
& \text{subject to} && \tau p_t+\left(T-\tau\right)p_d = \mathcal{P}T, \\
& && p_t \geq 0, p_d \geq 0. \\
\end{aligned}
\end{dcases*}
\end{equation}
Utilizing the first constraint in $\mathfrak{L}'$, the objective function reduces to
\begin{equation}\label{obj_func}
p_t p_d = \frac{\mathcal{P}T}{T-\tau}p_t-\frac{\tau}{T-\tau}p_t^2 .
\end{equation}
By differentiating \eqref{obj_func} with respect to $p_t$ and equating it to zero, we obtain
\begin{equation}\label{diff_zero}
\frac{\mathcal{P}T}{T-\tau}-\frac{2\tau}{T-\tau}p_t = 0 ,
\end{equation}
which has the optimal values
\begin{equation}\label{opt_value}
p_t^{\text{o}} = \frac{\mathcal{P}T}{2\tau} ,
\end{equation}
and
\begin{equation}\label{opt_value2}
p_d^{\text{o}} = \frac{\mathcal{P}T}{2\left(T-\tau \right)} .
\end{equation}
We note that the optimal values in \eqref{opt_value} and \eqref{opt_value2} satisfy the second and the third constraints in $\mathfrak{L}'$.

\subsection{Proof of Theorem \ref{theo3}} \label{proof of Theorem3}
We denote the ratio of the jamming power budget allocated to the pilot phase by $\zeta$ ($0 \leq \zeta \leq 1$) so that
\begin{equation}\label{PowBudg_Frac2}
q_t = \frac{\zeta \mathcal{P}_wT}{\tau} ~ \text{and} ~ q_d = \frac{\left(1-\zeta \right) \mathcal{P}_wT}{\left(T-\tau \right)} .
\end{equation}
Since the jamming transmit powers in \eqref{PowBudg_Frac2} satisfy all constraints in \eqref{optimization user1}, $\mathfrak{J}$ is equivalent to
\begin{equation}\label{optimization user3}
\mathfrak{J}':
\begin{dcases*}
\begin{aligned}
& \underset{\zeta}{\text{minimize}}
&& \rho_m^{\star}\left(\zeta \right) \\
& \text{subject to} && 0 \leq \zeta \leq 1, \\
\end{aligned}
\end{dcases*}
\end{equation}
by substituting \eqref{PowBudg_Frac2} into the objective function in \eqref{optimization user1}. It is obvious that the constraint in $\mathfrak{J}'$ results in a convex feasible set. Furthermore, we use the second-order conditions \cite{BookBoyd} to indicate that the objective function is convex. The second derivative of $\rho_m^{\star}\left(\zeta \right)$ is equal to
\begin{equation}\label{SecDerivative}
\frac{\partial^2\rho_m^{\star}\left(\zeta \right)}{\partial\zeta^2} = \tau p_t p_d \matr{r}_m^H \matr{U}_m\left(\zeta \right) \matr{r}_m ,
\end{equation}
where
\begin{equation}\label{Mat_U_m}
\matr{U}_m\left(\zeta \right) \triangleq \\
\boldsymbol{\Gamma}_m^{-1}\left(\zeta \right)\left(2\boldsymbol{\Lambda}_m\left(\zeta \right) \boldsymbol{\Gamma}_m^{-1}\left(\zeta \right) \boldsymbol{\Lambda}_m\left(\zeta \right)-\boldsymbol{\Upsilon}_m\left(\zeta \right) \right)\boldsymbol{\Gamma}_m^{-1}\left(\zeta \right) .
\end{equation}
The matrices $\boldsymbol{\Gamma}_m\left(\zeta \right)$, $\boldsymbol{\Lambda}_m\left(\zeta \right)$ and $\boldsymbol{\Upsilon}_m\left(\zeta \right)$ are given by
\begin{equation}\label{CovMat1}
\boldsymbol{\Gamma}_m\left(\zeta \right) \triangleq \frac{\zeta \left(1-\zeta \right) \mathcal{P}_w^2T^2}{(T-\tau)}\boldsymbol{\Gamma}_w+\zeta \mathcal{P}_wT\boldsymbol{\Lambda}_w+\frac{\left(1-\zeta \right) \mathcal{P}_wT}{(T-\tau)}\boldsymbol{\Psi}_m \\
+\left(\tau p_t\matr{R}_m+\matr{I}_M \right)^T \otimes \left(p_d \sum_{i=1}^{K}\matr{R}_i+\matr{I}_M\right) ,
\end{equation}
\begin{equation}\label{CovMat2}
\boldsymbol{\Lambda}_m\left(\zeta \right) \triangleq \frac{\partial \boldsymbol{\Gamma}_m\left(\zeta \right)}{\partial \zeta} = \frac{\left(1-2\zeta \right)\mathcal{P}_w^2T^2}{(T-\tau)}\boldsymbol{\Gamma}_w \\
+\mathcal{P}_wT\boldsymbol{\Lambda}_w-\frac{\mathcal{P}_wT}{(T-\tau)}\boldsymbol{\Psi}_m ,
\end{equation}
\begin{equation}\label{CovMat3}
\boldsymbol{\Upsilon}_m\left(\zeta \right) \triangleq \frac{\partial^2 \boldsymbol{\Gamma}_m\left(\zeta \right)}{\partial \zeta^2} = \frac{-2\mathcal{P}_w^2T^2}{(T-\tau)}\boldsymbol{\Gamma}_w ,
\end{equation}
where $\boldsymbol{\Gamma}_w \triangleq \matr{R}_w^T \otimes \matr{R}_w+\matr{r}_w \matr{r}_w^H$, $\boldsymbol{\Lambda}_w \triangleq \matr{R}_w^T \otimes (p_d \sum_{i=1}^{K}\matr{R}_i+\matr{I}_M)$, and $\boldsymbol{\Psi}_m \triangleq (\tau p_t\matr{R}_m+\matr{I}_M)^T \otimes \matr{R}_w$. Regarding \eqref{CovMat1}, $\boldsymbol{\Gamma}_m\left(\zeta\right)$ is composed of the channel covariance matrices that are positive definite. Hence, $\boldsymbol{\Gamma}_m\left(\zeta\right)$ as well as $\boldsymbol{\Gamma}_m^{-1}\left(\zeta\right)$ are the positive definite matrices and $\boldsymbol{\Upsilon}_m\left(\zeta\right)$ is a negative definite matrix (or equivalently $-\boldsymbol{\Upsilon}_m\left(\zeta\right)$ is a positive definite matrix). Finally, following from \eqref{Mat_U_m}, we can conclude that $\matr{U}_m\left(\zeta\right)$ is a positive definite matrix and
\begin{equation}\label{SecDerivative2}
\frac{\partial^2\rho_m^{\star}\left(\zeta \right)}{\partial\zeta^2} > 0 .
\end{equation}

\subsection{Proof of Corollary \ref{cor2}} \label{proof of Corollary2}
We can rewrite \eqref{opt_sjnr_r2_1} as
\begin{align}\label{opt_sjnr_r2_2}
\rho_m^{\star} &= \tau p_t p_d \matr{r}_m^H \left(\tau q_t \matr{R}_{w}^{T} \otimes q_d\matr{R}_{w} + \tau q_tq_d \matr{r}_w \matr{r}_w^{H} \right)^{-1} \matr{r}_m \nonumber \\
&= \frac{p_tp_d}{q_tq_d} \matr{r}_m^H \left(\matr{R}_{w}^{T} \otimes \matr{R}_{w} + \matr{r}_w \matr{r}_w^{H} \right)^{-1} \matr{r}_m .
\end{align}
$\rho_m^{\star}$ in \eqref{opt_sjnr_r2_2} is a decreasing function of the multiplication $q_tq_d$. Therefore, $\mathfrak{J}$ is
\begin{equation}\label{optimization user4}
\mathfrak{J}'':
\begin{dcases*}
\begin{aligned}
& \underset{q_t, q_d}{\text{maximize}}
&& q_t q_d \\
& \text{subject to} && \tau q_t+\left(T-\tau\right)q_d = \mathcal{P}_wT, \\
& && q_t \geq 0, q_d \geq 0 . \\
\end{aligned}
\end{dcases*}
\end{equation}
The continuation of the proof follows similar lines as the proof of Corollary \ref{cor1}.

\end{document}